\documentclass[letterpaper, 10 pt, conference]{ieeeconf}
\IEEEoverridecommandlockouts
\usepackage{cite}
\usepackage{amsmath,amssymb,amsfonts}
\usepackage{algorithmic}
\usepackage{graphicx}
\usepackage{textcomp}
\usepackage{xcolor}
\usepackage{mathrsfs}
\usepackage{psfrag}
\usepackage{cancel}
\usepackage{verbatim}

\usepackage{circuitikz}
\usetikzlibrary{decorations.pathreplacing,calligraphy,patterns}

\makeatletter
\let\NAT@parse\undefined
\makeatother
\usepackage{hyperref}
\usepackage{theoremref}

\def\mc{\mathcal}
\def\mb{\mathbb}

\begin{document}
\sloppy
\title{\LARGE \bf Digital control of negative imaginary systems: a discrete-time hybrid integrator-gain system approach}
\author{Kanghong Shi,$\quad$Ian R. Petersen, \IEEEmembership{Life Fellow, IEEE} 
\thanks{This work was supported by the Australian Research Council under grants DP190102158 and DP230102443.}
\thanks{K. Shi and I. R. Petersen are with the School of Engineering, College of Engineering, Computing and Cybernetics, Australian National University, Canberra, Acton, ACT 2601, Australia.
        {\tt kanghong.shi@anu.edu.au}, {\tt ian.petersen@anu.edu.au}.}%
}

\newtheorem{definition}{Definition}
\newtheorem{theorem}{Theorem}
\newtheorem{conjecture}{Conjecture}
\newtheorem{lemma}{Lemma}
\newtheorem{remark}{Remark}
\newtheorem{corollary}{Corollary}
\newtheorem{assumption}{Assumption}

\maketitle

\thispagestyle{empty}
\pagestyle{empty}

\begin{abstract}
A hybrid integrator-gain system (HIGS) is a control element that switches between an integrator and a gain, which overcomes some inherent limitations of linear controllers. In this paper, we consider using discrete-time HIGS controllers for the digital control of negative imaginary (NI) systems. We show that the discrete-time HIGS themselves are step-advanced negative imaginary systems. For a minimal linear NI system, there always exists a HIGS controller that can asymptotically stablize it. An illustrative example is provided, where we use the proposed HIGS control method to stabilize a discrete-time mass-spring system.
\end{abstract}


\section{INTRODUCTION}
Hybrid integrator-gain systems (HIGS) are hybrid control elements introduced in \cite{deenen2017hybrid} to overcome fundamental limitations of linear time-invariant (LTI) control systems \cite{middleton1991trade,freudenberg2000survey}. A HIGS switches between an integrator mode and a gain mode so that a certain sector constraint is satisfied. To be specific, a HIGS is primarily designed to operate as an integrator, and it switches to the gain mode when its integrator dynamics tend to violate the sector constraint. The describing function of a HIGS has a phase lag of only $38.15$ degrees, which is much smaller than the $90$ degrees phase lag of an integrator. Reset elements including the Clegg integrators \cite{Clegg_1958} and first-order reset elements \cite{horowitz1975non,chait2002horowitz} also have such advantages. However, they generate discontinuous control signals which may cause chattering and degrade the system performance \cite{bartolini1989chattering}, while HIGS generate continuous control signals. HIGS controllers have attracted attention since it was introduced (e.g., see \cite{van2020hybrid,deenen2021projection,van2021frequency,van2021overcoming,van2023small,heemels2023existence}) and have found application on wafer scanners \cite{heertjes2023overview} and atomic force microscopy \cite{shi2022negative}, where the latter work was motivated by the negative imaginary property of HIGS.

Negative imaginary (NI) systems theory was introduced by Lanzon and Petersen in \cite{lanzon2008stability} and \cite{petersen2010feedback}, and has attracted attention from many control theorists \cite{xiong2010negative,song2012negative,mabrok2014generalizing,wang2015robust,bhowmick2017lti}. A typical example of NI systems is a mechanical system with colocated force actuators and position sensors. Motivated by the robust control of flexible structures \cite{preumont2018vibration,halim2001spatial,pota2002resonant}, which have highly resonant dynamics, NI systems theory uses positive position feedback control. Roughly speaking, a square real-rational proper transfer matrix $F(s)$ is said to be NI if has no poles on the open right half-plane and $j(F(j\omega)-F(j\omega)^*)\geq 0$ for all $\omega \geq 0$. The Nyquist plot of a single-input single-output (SISO) NI system is contained in the lower half of the complex plane. Under mild assumptions, an NI system $F(s)$ can be asymptotically stabilized using a strictly negative imaginary (SNI) system $F_s(s)$ in positive feedback if and only if the DC loop gain has all its eigenvalues less than unity; i.e., $\lambda_{max}(F(0)F_s(0))<1$. Compared with passivity theory which can deal with systems having relative degree of zero and one \cite{brogliato2007dissipative}, NI systems theory can deal with systems having relative degree of zero, one and two \cite{shi2024necessary}. NI systems theory has been applied in many fields including nano-positioning control \cite{mabrok2013spectral,das2014mimo,das2014resonant,das2015multivariable}, the control of lightly damped structures \cite{cai2010stability,rahman2015design,bhikkaji2011negative}, and the control of power systems \cite{chen2023nonlinear}, etc.

NI systems theory was extended to nonlinear systems in \cite{ghallab2018extending,shi2021robust,shi2023output}. Roughly speaking, a system is said to be nonlinear NI if it has a positive semidefinite storage function $V(x)$ such that $\dot V(x)\leq u^T\dot y$, where $x$, $u$ and $y$ are the state, input and output of the system, respectively. Under some assumptions, a nonlinear NI system can be stabilized using another nonlinear NI system with a certain strictness property; e.g., output strictly negative imaginary systems \cite{shi2023output}, or weakly strictly negative imaginary systems \cite{ghallab2018extending}. It is shown in \cite{shi2022negative} that a HIGS controller is a nonlinear NI system. Also, for any minimal SISO linear NI system, there exists a HIGS controller such that their closed-loop interconnection is asymptotically stable. Motivated by the effectiveness of HIGS in the control of NI systems, the paper \cite{shi2023MEMS} showed the nonlinear NI property of two variants of HIGS including the multi-HIGS which was introduced in \cite{Achten_HIGS_Skyhook_thesis_2020}, and the cascade of two HIGS. It was also proved in \cite{shi2023MEMS} that these two variants of HIGS controllers can be used in stabilizing linear NI systems. This stability result was then applied on a MEMS nanopositioner \cite{shi2023MEMS}.

However, although the use of HIGS as NI controllers follows from the stability analysis in continuous time, the control of physical systems often requires construction of digital controllers. For the purpose of digital control, a discrete-time HIGS was introduced in \cite{sharif2022discrete}, which has a similar working mechanism as the continuous-time HIGS. Meanwhile, a novel discrete-time NI systems definition was introduced in \cite{shi2023discrete}, which characterizes the dissipativity property for a ZOH sampled continuous-time NI system. Note that the discrete-time NI systems definition in \cite{shi2023discrete} is different from the previously introduced definition in \cite{ferrante2017discrete}, which was mapped from the continuous-time NI systems definition using a bilinear transform. Since the definition of discrete-time NI systems in \cite{shi2023discrete} is obtained using ZOH sampling, it is guaranteed to be satisfied by any ZOH sampled physical plant with the NI property. It is shown in \cite{shi2023discrete} that the closed-loop interconnection of a discrete-time NI system and a so-called step-advanced negative imaginary (SANI) system is asymptotically stable, given that either of the systems has some strictness property.

In this paper, we use discrete-time HIGS as controllers for NI systems. We show that a discrete-time HIGS is an SANI system. Furthermore, we establish the following stability result: for any discrete-time NI system, there exists a discrete-time HIGS controller that ensures closed-loop asymptotic stability. An illustrative example is provided, where a ZOH sampled mass-spring system is asymptotically stabilized using a HIGS controller. This paper contributes in providing a specific digital control framework for physical systems with the NI property. The implementation process of a HIGS controller only involves the selection of parameters in order to satisfy a simple condition. In comparison to the continuous-time design approach where a continuous-time controller is constructed based on the continuous-time model of the plant and subsequently discretized \cite{nevsic1999sufficient}, the advantages of the framework in the present paper are two-fold: (a) the design and implementation processes are simpler; (b) closed-loop stability is more rigorously guaranteed.

The rest of the paper is organized as follows. Section \ref{sec:pre} provides preliminary definitions and lemmas for discrete-time NI systems that are introduced in \cite{shi2023discrete}. Section \ref{sec:pre} also provides the state-space model of a discrete-time HIGS. Section \ref{sec:main} contains the main results of this paper. We show in Section \ref{sec:main} the NI property of a discrete-time HIGS. We also show that given a linear discrete-time NI plant, there always exists a HIGS controller that can stabilize the NI plant. An example is provided in Section \ref{sec:example}, where a discrete-time mass-spring system is stabilized by a HIGS controller using the proposed control framework. Section \ref{sec:conclusion} concludes the paper and discusses potential future work.

Notation: The notation in this paper is standard. $\mathbb R$ denotes the field of real numbers. $\mb N$ denotes the set of nonnegative integers. $\mathbb R^{m\times n}$ denotes the space of real matrices of dimension $m\times n$. $A^T$ denotes the transpose of a matrix $A$.  $A^{-T}$ denotes the transpose of the inverse of $A$; that is, $A^{-T}=(A^{-1})^T=(A^T)^{-1}$. $\lambda_{max}(A)$ denotes the largest eigenvalue of a matrix $A$ with real spectrum. $\|\cdot\|$ denotes the standard Euclidean norm. For a real symmetric or complex Hermitian matrix $P$, $P>0\ (P\geq 0)$ denotes the positive (semi-)definiteness of a matrix $P$ and $P<0\ (P\leq 0)$ denotes the negative (semi-)definiteness of a matrix $P$. A function $V: \mb R^n \to \mb R$ is said to be positive definite if $V(0)=0$ and $V(x)>0$ for all $x\neq 0$.

\section{PRELIMINARIES}\label{sec:pre}
\subsection{Discrete-time NI systems}
Consider the system
\begin{subequations}\label{eq:DT_nonlinear}
\begin{align}
	x_{k+1} =&\ f(x_k,u_k),\label{eq:state eq}\\
	y_k=&\ h(x_k),\label{eq:output eq}
\end{align}	
\end{subequations}
where $f:\mathbb R^n \times \mathbb R^p\to \mathbb R^n$ and $h:\mathbb R^n \to \mathbb R^p$. Here $u_k,y_k \in \mathbb R^p$ and $x_k\in \mathbb R^n$ are the input, output and state of the system at time step $k\in \mathbb N$, respectively.

\begin{definition}\label{def:DT_NNI}\cite{shi2023discrete}
The system (\ref{eq:DT_nonlinear}) is said to be a discrete-time negative imaginary (NI) system if there exists a positive definite function $V\colon \mb R^n \to \mb R$ such that for arbitrary $x_k$ and $u_k$,
\begin{equation}\label{eq:NNI ineq}
V(x_{k+1})-V(x_{k})\leq u_k^T\left(y_{k+1}-y_{k}\right),	
\end{equation}
for all $k$.
\end{definition}

We provide the necessary and sufficient linear matrix inequalities (LMI) conditions under which Definition \ref{def:DT_NNI} is satisfied by a linear system of the form
\begin{subequations}\label{eq:G(z)}
	\begin{align}
\Sigma\colon\ 		x_{k+1} =&\ Ax_k+Bu_k,\label{eq:G(z) state eq}\\
		y_k =&\ Cx_k,
	\end{align}
\end{subequations}
where $x_k\in \mathbb R^n$, $u_k,y_k\in \mb R^p$ are the system state, input and output, respectively.
\begin{lemma}\label{lemma:LMI new DT-NI}\cite{shi2023discrete}
Suppose the linear system (\ref{eq:G(z)}) satisfies $\det(I-A)\neq 0$. Then the system (\ref{eq:G(z)}) is NI with a positive definite quadratic storage function satisfying (\ref{eq:NNI ineq}) if and only if there exists a real matrix $P=P^T>0$ such that
\begin{equation*}
	A^TPA-P\leq 0 \quad \textnormal{and} \quad C = B^T(I-A)^{-T}P.
\end{equation*}
\end{lemma}

We present in the following, the definition of SANI systems. Consider the system 
\begin{subequations}\label{eq:nonlinear SANI system}
\begin{align}
\widetilde x_{k+1} =&\ \widetilde f(\widetilde x_k, \widetilde u_k),\\
	\widetilde y_k=&\ \widetilde h(\widetilde x_k, \widetilde u_k),
\end{align}	
\end{subequations}
where $\widetilde f:\mathbb R^n \times \mathbb R^p\to \mathbb R^n$ and $\widetilde h:\mathbb R^n \to \mathbb R^p$. Here $\widetilde u, \widehat y \in \mathbb R^p$ and $\widetilde x\in \mathbb R^n$ are the input, output and state of the system at time step $k\in \mathbb N$, respectively.
\begin{definition}\label{def:SANI}\cite{shi2023discrete}
	The system \eqref{eq:nonlinear SANI system} is said to be a step-advanced negative imaginary (SANI) system if there exists a function $\widehat h(x_k)$ such that:
\begin{enumerate}
	\item $\widetilde h(\widetilde x_k, \widetilde u_k)=\widehat h(\widetilde f(\widetilde x_k, \widetilde u_k))$;
	\item there exists a positive definite function $\widetilde V\colon \mb R^n \to \mb R$ such that for arbitrary state $\widetilde x_k$ and input $\widetilde u_k$,
\begin{equation*}
\widetilde V(\widetilde x_{k+1})-\widetilde V(\widetilde x_{k})\leq \widetilde u_k^T\left(\widehat h(\widetilde x_{k+1})-\widehat h(\widetilde x_{k})\right)
\end{equation*}
for all $k$.
\end{enumerate}
\end{definition}
\begin{remark}\label{remark:NI_SANI}
	Definition \ref{def:SANI} can be regarded as a variant of Definition \ref{def:DT_NNI} in a way such that the system output takes one step advance. To be specific, suppose the system (\ref{eq:DT_nonlinear}) is NI as per Definition \ref{def:DT_NNI}. Then a system with the same state equation (\ref{eq:state eq}) and an output equation $\widetilde y_k = h(x_{k+1})=h(f(x_k,u_k))$ is an SANI system. Note that this does not affect the causality of the system because $h(f(x_k,u_k))$ is a function of the state $x_k$ and input $u_k$ of the current step $k$.
\end{remark}

\subsection{Discrete-time hybrid integrator-gain systems}
Discrete-time HIGS were introduced in \cite{sharif2022discrete}. We adapt the model in \cite{sharif2022discrete} to fit the system model (\ref{eq:DT_nonlinear}) in the following.
	\begin{equation}\label{eq:DT-HIGS}
		\mathcal{H}:
		\begin{cases}
			x_h(k+1) = x_h(k)+\omega_h e(k), & \text{if}\, (x_h(k),e(k)) \in \mathcal{F}\\
			x_h(k+1) = k_he(k), & \text{if}\, (x_h(k),e(k)) \notin \mathcal{F}\\
			y_h(k) = x_h(k+1).
		\end{cases}
	\end{equation}
Here, $e(k),x_h(k),y_h(k)\in \mb R$ are the system input, state and output, respectively. The constant parameters $\omega_h\geq 0$ and $k_h>0$ are called the integrator frequency and the gain value, respectively. The HIGS is designed to operate under the sector constraint $(x_h(k),e(k)) \in \mathcal{F}$, where $\mc F$ is given by
\begin{align}\label{eq:F}
	\mc F =& \{(x_h(k),e(k))\in \mb R^2\mid \notag\\
	 &(x_h(k)+\omega_h e(k))e(k)\geq \frac{1}{k_h}(x_h(k)+\omega_h e(k))^2\}.
\end{align}
At time step $k$, if $(e(k),y_h(k))\in \mc F$, then $(e(k),y_h(k))$ is contained in the sector $[0,k_h]$.
The HIGS is designed to operate primarily in the integrator mode if the input $e(k)$ leads to an output $y_h(k)$ within the sector $[0,k_h]$ under the integrator mode dynamics. Otherwise, the system operates in the gain mode so that $y_h(k)=k_he(k)$, which automatically satisfies the sector constraint $[0,k_h]$. According to (\ref{eq:DT-HIGS}), regardless of the initial condition $x_h(0)$, the discrete-time HIGS will remain in the sector given in $\mc F$ from the time step $k=1$. In what follows, we denote $e(k)$, $x_h(k)$ and $y_h(k)$ by $e_k$, $\widetilde x_k$ and $\widetilde y_k$ respectively for convenience. Note that the parameter $\omega_h$ in the present paper corresponds to the product $\omega_h T_s$ in \cite{sharif2022discrete}, where $\omega_h$ is the integrator frequency of the corresponding continuous-time integrator and $T_s$ is the sampling period. Since we only consider the discrete-time case in the present paper, we regard $\omega_h$ as the discrete-time integrator frequency.  

\section{MAIN RESULTS}\label{sec:main}
\subsection{SANI property of the HIGS}
We show in the following that the HIGS given in (\ref{eq:DT-HIGS}) is an SANI system. 
\begin{theorem}
The system given in (\ref{eq:DT-HIGS}) is an SANI system with the storage function
\begin{equation}\label{eq:HIGS storage function}
	\widetilde V(\widetilde x_k) = \frac{1}{2k_h}\widetilde x_k^2
\end{equation}
satisfying
\begin{equation}\label{eq:NNI ineq for HIGS aux}
	\widetilde V(\widetilde x_{k+1})-\widetilde V(\widetilde x_k)\leq e_k(\widetilde x_{k+1}-\widetilde x_k),
\end{equation}
for any input $e_k$ and state $\widetilde x_k$.
\end{theorem}
\begin{proof}
According to Definition \ref{def:SANI} and Remark \ref{remark:NI_SANI}, the HIGS is an SANI system if it is NI from the input $e_k$ to the state $\widetilde x_k$. Hence, we prove in the following that (\ref{eq:NNI ineq for HIGS aux}) is satisfied in both integrator mode and gain mode. Substituting (\ref{eq:HIGS storage function}) into (\ref{eq:NNI ineq for HIGS aux}) yields
\begin{equation}\label{eq:NNI HIGS aux expansion}
	\frac{1}{2k_h}\widetilde x_{k+1}^2-\frac{1}{2k_h}\widetilde x_k^2	\leq e_k (\widetilde x_{k+1}-\widetilde x_k),
	\end{equation}
which is required to be satisfied in both modes.

\noindent\textit{\textbf{Case 1}}. In the integrator mode, we have the state equation $\widetilde x_{k+1} = \widetilde x_k+\omega_h e_k$ and also $(\widetilde x_k,e_k) \in \mathcal{F}$. In this case, (\ref{eq:NNI HIGS aux expansion}) becomes
\begin{equation}\label{eq:integrator mode NI property}
	2\widetilde x_k e_k\leq (2k_h-\omega_h)e_k^2,
\end{equation}
which is always satisfied when $e_k=0$. When $e_k\neq 0$, (\ref{eq:integrator mode NI property}) can be rewritten as
\begin{equation}\label{eq:HIGS integrator NI ineq target}
	2 \frac{\widetilde x_k}{e_k}\leq 2k_h-\omega_h.
\end{equation}
The condition $(\widetilde x_k,e_k) \in \mathcal{F}$ implies
\begin{equation*}
	\widetilde x_k^2+(2\omega_h-k_h)\widetilde x_ke_k+(\omega_h-k_h\omega)e_k^2 \leq 0.
\end{equation*}
This implies that for $e_k\neq 0$,
\begin{equation}\label{eq:HIGS integrator F meaning}
\left( \frac{\widetilde x_k}{e_k}\right)^2+(2\omega_h-k_h)\frac{\widetilde x_k}{e_k}+(\omega_h^2-k_h\omega_h) \leq 0.	
\end{equation}
By solving (\ref{eq:HIGS integrator F meaning}), we have that operating in the integrator mode requires the HIGS input $e_k$ and state $\widetilde x_k$ to satisfy
\begin{equation*}
	-\omega_h\leq \frac{\widetilde x_k}{e_k}\leq k_h-\omega_h.
\end{equation*}
Such a pair of $\widetilde x_k$ and $e_k$ always satisfies (\ref{eq:HIGS integrator NI ineq target}).

\noindent\textit{\textbf{Case 2}}. In the gain mode, we have that $\widetilde x_{k+1} = k_h e_k$ and $(\widetilde x_k,e_k) \notin \mathcal{F}$. In this case, (\ref{eq:NNI HIGS aux expansion}) becomes
\begin{equation*}
	\widetilde x_k^2-2k_h\widetilde x_ke_k+k_h^2e_k^2\geq 0,
\end{equation*}
which always holds because
\begin{equation*}
	\widetilde x_k^2-2k_h\widetilde x_ke_k+k_h^2e_k^2= (\widetilde x_k-k_he_k)^2\geq 0.
\end{equation*}
Since condition (\ref{eq:NNI ineq for HIGS aux}) is satisfied in both modes, then the system (\ref{eq:DT-HIGS}) is an SANI system.
\end{proof}

\begin{figure}[h!]
\centering
\psfrag{in_1}{$u_k$}
\psfrag{y_1}{$y_k$}
\psfrag{e}{$e_k$}
\psfrag{x_h}{$\widetilde y_k$}
\psfrag{plant}{\hspace{-0.1cm}$G(z)$}
\psfrag{HIGS}{\hspace{-0.3cm}$ HIGS\ {\mc H}$}
\includegraphics[width=8.5cm]{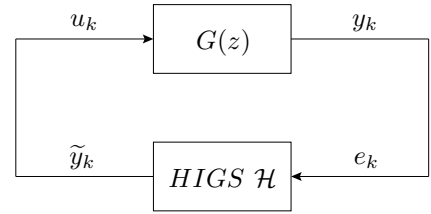}
\caption{Closed-loop interconnection of the system (\ref{eq:G(z)}) with the transfer matrix $G(z)$ and the HIGS $\mc H$ given in (\ref{eq:DT-HIGS}).}
\label{fig:interconnection}
\end{figure}

\subsection{Stability for the interconnection of a linear NI system and a HIGS}
Motivated by the SANI property of the HIGS, we investigate whether a HIGS controller can be applied in the control of a minimal SISO linear NI system. Consider a SISO system of the form (\ref{eq:G(z)}) with $u_k,y_k\in\mb R$, which has a transfer function matrix $G(z)$. We show in the following that if the system $\Sigma$ in (\ref{eq:G(z)}) is NI, then there exists a HIGS controller $\mc H$ such that the positive feedback interconnection of $\Sigma$ and $\mc H$ shown in Fig.~\ref{fig:interconnection} is asymptotically stable. The setting of the interconnection can be described as follows:
\begin{align*}
e_k =&\  y_k;\\
u_k =&\ \widetilde  y_k.
\end{align*}
This means the HIGS $\mc H$ takes the output of the system $\Sigma$ as its input and feeds back its output to the system $\Sigma$ as its input.

\begin{theorem}\label{theorem:stability}
	Suppose the SISO minimal system (\ref{eq:G(z)}) with transfer function matrix $G(z)$ is NI and satisfies $\det(I-A)\neq 0$. Suppose the HIGS $\mc H$ of the form (\ref{eq:DT-HIGS}) satisfies $0<\omega_h\leq k_h<\frac{1}{G(1)}$. Then the closed-loop interconnection of $G(z)$ and $\mc H$ as shown in Fig.~\ref{fig:interconnection} is asymptotically stable.
\end{theorem}
\begin{proof}
According to Lemma \ref{lemma:LMI new DT-NI}, the minimal system (\ref{eq:G(z)}) is NI if and only if there exists a matrix $P=P^T>0$ such that
\begin{equation*}
	A^TPA-P\leq 0,\quad \textnormal{and}\quad C = B^T(I-A)^{-T}P.
\end{equation*}
We construct the following Lyapunov function for the closed-loop interconnection:
\begin{align*}
	W(x_k,\widetilde x_k) =&\ V(x_k)+\widetilde V(\widetilde x_k)-Cx_k\widetilde x_k\notag\\
	=&\ \frac{1}{2}x_k^TPx_k+\frac{1}{2k_h}x_k^2-Cx_k\widetilde x_k.
\end{align*}
Rewriting this as a quadratic form, we have that 
\begin{equation*}
	W(x_k,\widetilde x_k) = \frac{1}{2}\begin{bmatrix}
		x_k^T & \widetilde x_k
	\end{bmatrix}\begin{bmatrix}
		P & -C^T \\ -C & \frac{1}{k_h}
	\end{bmatrix}\begin{bmatrix}
		x_k \\ \widetilde x_k
	\end{bmatrix}.
\end{equation*}
Using the Schur complement theorem, to ensure that $W(x_k,\widetilde x_k)$ is positive definite, we need
\begin{equation}\label{eq:schur complement ineq}
	\frac{1}{k_h}-CP^{-1}C^T>0.
\end{equation}
Since $C=B^T(I-A)^{-T}P$, then (\ref{eq:schur complement ineq}) can be rewritten as
\begin{equation*}
	\frac{1}{k_h}-C(I-A)^{-1}B>0,
\end{equation*}
which is satisfied because $G(1)=C(I-A)^{-1}B$ and
\begin{equation}\label{eq:k_h condition}
	k_hG(1)<1.
\end{equation}
Note that $G(1)\neq 0$ according to the positive definiteness of $P$ and the fact that $C$ is not a zero row vector, which is guaranteed by the minimality of the system.
We use Lyapunov's direct method \cite{Kalman1960} in the following. Taking the difference between $W(x_{k+1},\widetilde x_{k+1})$ and $W(x_{k},\widetilde x_{k})$, we have
\begin{align}
	W&(x_{k+1},\widetilde x_{k+1})-W(x_k,\widetilde x_k)\notag\\
	=&\ V(x_{k+1})+\widetilde V(\widetilde x_{k+1})-Cx_{k+1}\widetilde x_{k+1}-V(x_k)-\widetilde V(\widetilde x_k)\notag\\
	&+Cx_k\widetilde x_k\notag\\
	\leq &\ u_k(y_{k+1}-y_k)+e_k(\widetilde x_{k+1}-\widetilde x_k)-Cx_{k+1}\widetilde x_{k+1}+Cx_k\widetilde x_k\notag\\
	 =&\ \widetilde x_{k+1}(e_{k+1}-e_k)+e_k(\widetilde x_{k+1}-\widetilde x_k)-e_{k+1}\widetilde x_{k+1}+e_k\widetilde x_k\notag\\
  =&\ 0.\label{eq:W difference}
\end{align}
which implies that the system is Lyapunov stable. Furthermore, $W(x_{k+1},\widetilde x_{k+1})-W(x_k,\widetilde x_k)=0$ only if the inequality in (\ref{eq:W difference}) is an equality. That is
\begin{align}
    V(x_{k+1})-V(x_k)=&\ u_k(y_{k+1}-y_k);\label{eq:plant lossless}\\
    \widetilde V(\widetilde x_{k+1})-\widetilde V(\widetilde x_k)=&\ e_k(\widetilde x_{k+1}-\widetilde x_{k}).\label{eq:HIGS lossless}
\end{align}
We prove in the following that (\ref{eq:plant lossless}) and (\ref{eq:HIGS lossless}) cannot hold together at all time indices $k$ unless $(x_k,\widetilde x_k)=(0,0)$. We consider the case that (\ref{eq:plant lossless}) and (\ref{eq:HIGS lossless}) hold for some index $k$ and all future indices $k+1$, $k+2$, $\cdots$.
When (\ref{eq:HIGS lossless}) holds, we have that
\begin{equation}\label{eq:HIGS lossless expansion}
    \frac{1}{2k_h}\widetilde x_{k+1}^2-\frac{1}{2k_h}\widetilde x_k^2 = e_k(\widetilde x_{k+1}-\widetilde x_k).
\end{equation}
We consider the following two cases, where the HIGS is assumed to work in the integrator mode and the gain mode, respectively.

\noindent\textit{\textbf{Case 1}}. Integrator mode. In this case, $(\widetilde x_k,e_k)\in \mc F$ and
\begin{equation}\label{eq:integrator mode state eq}
    \widetilde x_{k+1} = \widetilde x_k+\omega_h e_k.
\end{equation}
Substituting (\ref{eq:integrator mode state eq}) in (\ref{eq:HIGS lossless expansion}) yields
\begin{equation}\label{eq:integrator mode lossless}
    (\omega_h-2k_h)e_k^2+2\widetilde x_k e_k=0.
\end{equation}
\textit{\textbf{Case 1a}}. Suppose $e_k\neq 0$. Then we have $\widetilde x_k = (k_h-\frac{\omega_h}{2})e_k$, which can be substituted in the inequality in (\ref{eq:F}) and yields
\begin{equation*}
    (k_h+\frac{\omega_h}{2})e_k^2 \geq \frac{1}{k_h}(k_h+\frac{\omega_h}{2})^2e_k^2.
\end{equation*}
This, after simplification, becomes
\begin{equation*}
    \omega_h^2+2k_h\omega_h\leq 0.
\end{equation*}
Considering the fact that $k_h>0$ and $\omega>0$, the above condition can never be satisfied. Hence, \textit{Case 1a} can never happen.

\noindent\textit{\textbf{Case 1b}}. Suppose $e_k = 0$. Then (\ref{eq:integrator mode lossless}) is always satisfied. In this case, $(\widetilde x_k,e_k)\in \mathcal F$ implies that $\widetilde x_k=0$. According to (\ref{eq:integrator mode state eq}), we have that $\widetilde x_{k+1}=0$ as well. The condition for $(\widetilde x_{k+1},e_{k+1})\in \mc F$ can be simplified to be
\begin{equation*}
    (k_h-\omega_h)e_{k+1}^2\geq 0.
\end{equation*}
The fact that $k_h-\omega_h\geq 0$ guarantees that the next active mode is the integrator mode. Note that this condition is irrelevant to the HIGS input or state. Indeed, since \textit{Case 1a} can never happen, then the system will fall in \textit{Case 1b} for all future time indices $k+1$, $k+2$, $\cdots$. Following from a similar analysis, we have that 
\begin{equation}\label{eq:case1b e=0}
    0=e_k=e_{k+1}=e_{k+2}=\cdots,
\end{equation}
and also
\begin{equation}\label{eq:case1b tilde x=0}
    0=\widetilde x_{k} = \widetilde x_{k+1} = \widetilde x_{k+2}=\cdots.
\end{equation}
Since $u_k = \widetilde y_k = \widetilde x_{k+1}$, then according to (\ref{eq:case1b tilde x=0}) and (\ref{eq:G(z) state eq}), we have that
\begin{equation}\label{eq:case1b iteration}
    x_{k+1} = Ax_k,\quad x_{k+2} = Ax_{k+1} = A^2x_k, \quad \cdots
\end{equation}
Since $e_k = y_k = Cx_k$, then according to (\ref{eq:case1b e=0}) and (\ref{eq:case1b iteration}), we have that
\begin{equation*}
    \begin{bmatrix}
        C\\CA\\ \vdots\\ CA^{n-1}
    \end{bmatrix}x_k=0.
\end{equation*}
This implies that $x_k=0$ due to the observability of $G(z)$. In this case, $(x_k,\widetilde x_k)=(0,0)$. The closed-loop system is already in its equilibrium.

\noindent\textit{\textbf{Case 2}}. Gain mode. In this case, $(\widetilde x_k,e_k)\notin \mc F$, and we have that
\begin{equation}\label{eq:gain mode state eq}
    \widetilde x_{k+1}=k_he_k.
\end{equation}
Substituting (\ref{eq:gain mode state eq}) in (\ref{eq:HIGS lossless expansion}), we have that
\begin{equation*}
    (\widetilde x_k - k_h e_k)^2 = 0.
\end{equation*}
That is
\begin{equation}\label{eq:case2 tilde x_k=tilde x_{k+1}}
    \widetilde x_k = k_he_k = \widetilde x_{k+1}.
\end{equation}
The condition $(\widetilde x_k,e_k)\notin \mc F$ implies that
\begin{equation*}
    (k_h+\omega_h)e_k^2 > 0.
\end{equation*}
This implies that $e_k\neq 0$. We only need consider the case that the HIGS operates in the gain mode for all future indices. This is because that under the constraints (\ref{eq:HIGS lossless}), if it enters the integrator mode, it will never exit the integrator mode, according to the analysis in \textit{Case 1b}. Then it falls into \textit{Case 1}. In the case that the system keeps operating in the gain mode, following from the same derivation of (\ref{eq:case2 tilde x_k=tilde x_{k+1}}), we have that
\begin{equation}\label{eq:case2 tilde x_{k+1} = tilde x_{k+2}}
	\widetilde x_{k+1} = k_he_{k+1}=\widetilde x_{k+2}.
\end{equation}
Comparing (\ref{eq:case2 tilde x_k=tilde x_{k+1}}), (\ref{eq:case2 tilde x_{k+1} = tilde x_{k+2}}) and similar equations for future time indices, we have that
\begin{equation*}
	\widetilde x_k = k_he_k =\widetilde x_{k+1} = k_h e_{k+1} = \widetilde x_{k+2} = k_he_{k+2}=\cdots.
\end{equation*}
That is
\begin{equation*}
	e_k = e_{k+1}=e_{k+2}=\cdots.
\end{equation*}
This implies that
\begin{equation}\label{eq:case2 y_k all equal}
	y_k = y_{k+1}=y_{k+2}=\cdots.
\end{equation}
In this case, we have that
\begin{align*}
	x_{k+1} =&\ Ax_k+Bu_k = Ax_k+B\widetilde y_k=Ax_k+B\widetilde x_{k+1}\notag\\
	=&\ Ax_k+Bk_he_k = Ax_k+k_hBCx_k\notag\\
	=&\ (A+k_hBC)x_k.
\end{align*}
Similarly, we have
\begin{align*}
	x_{k+2}=&\ (A+k_hBC)x_{k+1}=(A+k_hBC)^2x_k,\notag\\
	\vdots &\notag\\
	x_{k+n-1}=&\ (A+k_hBC)^{n-1}x_k.
\end{align*}
According to (\ref{eq:case2 y_k all equal}), we have that
\begin{equation*}
	\begin{bmatrix}
		y_{k+1}-y_k\\ y_{k+2}-y_{k+1}\\ \vdots \\ y_{k+n}-y_{k+n-1}
	\end{bmatrix}=0,
\end{equation*}
which implies
\begin{equation}\label{eq:case2 observability condition}
	\begin{bmatrix}
		C\\C(A+k_hBC)\\ \vdots \\ C(A+k_hBC)^{n-1}
	\end{bmatrix}(x_{k+1}-x_k) = 0.
\end{equation}
We use eigenvector test to prove that observability of $(A,C)$ implies that of $(A+k_hBC,C)$. Suppose $\eta \neq 0$ is a vector in the kernel of $C$; i.e., $C\eta = 0$. Then it is not an eigenvector of $A$; i.e., $A\eta \neq \lambda \eta$ for all scalars $\lambda$. Then $\eta$ is not an eigenvector of $A+k_hBC$ as well because $(A+k_hBC)\eta = A\eta + k_hBC\eta = A\eta \neq \lambda \eta$ for all $\lambda$, considering $C\eta = 0$. Hence, $(A+k_hBC,C)$ is observable and (\ref{eq:case2 observability condition}) implies that $x_{k+1}=x_k$. That is, $x_k$ is an eigenvector of $A+k_hBC$ with an eigenvalue $\lambda=1$. This implies that
\begin{equation*}
	x_k = x_{k+1}=x_{k+2}=\cdots.
\end{equation*}
In this case, we also have that
\begin{align*}
	x_k =&\ x_{k+1}=Ax_k+Bu_k=Ax_k+B\widetilde y_k = Ax_k+B\widetilde x_{k+1}\notag\\
	=&\ Ax_k+Bk_he_k.
\end{align*}
This implies that
\begin{equation*}
	x_k=k_h(I-A)^{-1}Be_k.
\end{equation*}
Also, we have that
\begin{equation}\label{eq:e_k=k_hG(1)e_k}
	e_k = Cx_k=k_hC(I-A)^{-1}Be_k.
\end{equation}
Since we have $e_k\neq 0$ in \textit{Case 2}, then (\ref{eq:e_k=k_hG(1)e_k}) implies that
\begin{equation*}
	k_hC(I-A)^{-1}B=1,
\end{equation*}
which is
\begin{equation*}
	k_hG(1)=1.
\end{equation*}
This contradicts (\ref{eq:k_h condition}).
To conclude, we have shown that if (\ref{eq:plant lossless}) and (\ref{eq:HIGS lossless}) hold together for all future time indices, then the HIGS cannot stay in the gain mode according to the analysis in \textit{Case 2}. It will eventually switch to the integrator mode. Then, according to the analysis in \textit{Case 1}, the HIGS will stay in the integrator mode. However, we have shown in \textit{Case 1b} that this is only possible if the system is already at the equilibrium. In other words, if the system is not at the equilibrium, then (\ref{eq:plant lossless}) and (\ref{eq:HIGS lossless}) cannot hold together for all future indices, and $W(x_k,\widetilde x_k)$ will eventually decrease again until the system reaches its equilibrium. This means that the system is asymptotically stable.
\end{proof}

\section{EXAMPLE}\label{sec:example}
In this section, we demonstrate the feasibility of the proposed stability results in Theorem \ref{theorem:stability}. Consider a mass-spring system shown in Fig.~\ref{fig:example}.
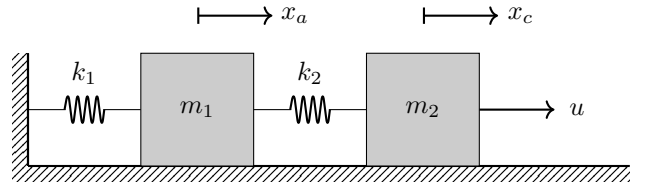
\begin{figure}[h!]
\centering
\ctikzset{bipoles/length=1cm}
\begin{circuitikz}
	\pattern[pattern = north east lines] (-0.2,-0.2) rectangle (0,1.5);
	\draw[thick] (0,0) -- (0,1.5);
	\pattern[pattern = north east lines] (0,-0.2) rectangle (8,0);
	\draw[thick] (0,0) -- (8,0);	
	\draw[thick] (0,0) -- (0,1.5);
	\draw (0,0.75) to [spring, l=${k_1}$] (1.5,0.75);
	\draw[fill=gray!40] (1.5,0) rectangle (3,1.5);
\node at (2.25,0.75){$m_1$};
\draw (3,0.75) to [spring, l=${k_2}$] (4.5,0.75);
	\draw[fill=gray!40] (4.5,0) rectangle (6,1.5);
\node at (5.25,0.75){$m_2$};
\draw[thick, ->] (6,0.75) -- (7,0.75);
        \node at (7.3,0.75){$u$};
 \draw[thick, |->] (2.25,2) -- (3.25,2);
        \node at (3.55,2){$x_a$};
\draw[thick, |->] (5.25,2) -- (6.25,2);
     \node at (6.55,2){$x_c$};
\end{circuitikz}
\caption{A mass-spring system with masses $m_1 = 0.04kg$, $m_2 = 0.02kg$ and spring constants $k_1=2N/m$ and $k_2 = 1N/m$. $x_a$ and $x_c$ denote the displacement of the masses $m_1$ and $m_2$, respectively. A force input $u$ is applied on the mass $m_2$.}
\label{fig:example}
\end{figure}

Sampling the system with the period $h=0.04s$ using a ZOH device, we obtain the following discrete-time model (see also \cite{aastrom2013computer}).
\begin{subequations}\label{eq:example plant}
	\begin{align}
\Sigma:\		x_{k+1} =&\ A x_k+Bu_k,\\
		y_k =&\ Cx_k,
	\end{align}
\end{subequations}
where
\begin{align}
	A =& \scriptsize \begin{bmatrix}
			\frac{1}{3}c_1+\frac{2}{3}c_2 & \frac{1}{15}s_1+\frac{1}{15}s_2 & \frac{1}{3}c_1-\frac{1}{3}c_2 & \frac{1}{15}s_1-\frac{1}{30}s_2 \\ \\ -\frac{5}{3}s_1-\frac{20}{3}s_2 & \frac{1}{3}c_1+\frac{2}{3}c_2 & -\frac{5}{3}s_1+\frac{10}{3}s_2 & \frac{1}{3}c_1-\frac{1}{3}c_2\\ \\  \frac{2}{3}c_1-\frac{2}{3}c_2 & \frac{2}{15}s_1-\frac{1}{15}s_2 & \frac{2}{3}c_1+\frac{1}{3}c_2 & \frac{2}{15}s_1+\frac{1}{30}s_2\\ \\-\frac{10}{3}s_1+\frac{20}{3}s_2 & \frac{2}{3}c_1-\frac{2}{3}c_2 & -\frac{10}{3}s_1-\frac{10}{3}s_2 & \frac{2}{3}c_1+\frac{1}{3}c_2
		\end{bmatrix},\notag\\
	B=& \begin{bmatrix}
			-\frac{2}{3}c_1+\frac{1}{6}c_2+\frac{1}{2} \\ \frac{10}{3}s_1-\frac{5}{3}s_2\\ -\frac{4}{3}c_1-\frac{1}{6}c_2+\frac{3}{2}  \\ \frac{20}{3}s_1-\frac{5}{3}s_2
		\end{bmatrix}\notag\\
	C =& \begin{bmatrix}
		0& 0 & 1 & 0
		\end{bmatrix},\notag
\end{align}
with
\begin{align*}
	c_1 =& \cos(5h)=\cos(0.2); \ c_2 = \cos(10h)=\cos(0.4);\\
	s_1 =& \sin(5h)=\sin(0.2); \hspace{0.188cm} s_2 = \sin(10h)=\sin(0.4).
\end{align*}
Here, $x_k = \begin{bmatrix}
	x_{ak} & x_{bk} & x_{ck} &x_{dk}
\end{bmatrix}^T\in \mb R^4$, $u_k,y_k\in \mb R$ are the state, input and output of the system, respectively. $x_{ak}$ and $x_{bk}$ are the displacement and velocity of the mass $m_1$ while $x_{ck}$ and $x_{dk}$ are the displacement and velocity of the mass $m_2$, respectively, at time step $k$. This system is NI according to Definition \ref{def:DT_NNI} with the storage function
\begin{equation*}
	V(x_k) = x_k^TPx_k,
\end{equation*}
where
\begin{equation*}
	P = \begin{bmatrix}
		3 & 0 & -1 & 0\\
		0 & 0.04 & 0 & 0\\
		-1 & 0 & 1 & 0\\
		0 & 0 & 0 & 0.02
	\end{bmatrix}.
\end{equation*}

We apply a HIGS controller of the form (\ref{eq:DT-HIGS}) in positive feedback with the plant (\ref{eq:example plant}). For the plant (\ref{eq:example plant}), we have that $G(1) = C(I-A)^{-1}B = \frac{3}{2}$. Hence, we choose the HIGS parameters to be $\omega_h = 0.1$, $k_h = 0.6$, which satisfies the condition $0<\omega_h\leq k_h<\frac{1}{G(1)}$ as required in Theorem \ref{theorem:stability}. A simulation is implemented with the initial values $x_0 = \begin{bmatrix}
	3 & -2 & 5 & -1
\end{bmatrix}^T$. The state trajectories of the plant and the HIGS controller are shown in Fig.~\ref{fig:simulation}.
\begin{figure}[h!]
\centering
\psfrag{HIGS state}{\scriptsize$x(k)$}
\psfrag{plant state}{\scriptsize$\widehat x(k)$}
\psfrag{k}{$k$}
\psfrag{xa}{\scriptsize$x_a$}
\psfrag{xb}{\scriptsize$x_b$}
\psfrag{xc}{\scriptsize$x_c$}
\psfrag{xd}{\scriptsize$x_d$}
\psfrag{xco}{\scriptsize$\widetilde x$}
\psfrag{State}{\small State}
\psfrag{State Trajectories}{\hspace{-0.3cm}State Trajectories}
\psfrag{5}{\scriptsize$5$}
\psfrag{10}{\scriptsize$10$}
\psfrag{15}{\scriptsize$15$}
\psfrag{0}{\scriptsize$0$}
\psfrag{500}{\scriptsize$500$}
\psfrag{1000}{\scriptsize$1000$}
\psfrag{1500}{\scriptsize$1500$}
\psfrag{2000}{\scriptsize$2000$}
\psfrag{6}{\scriptsize$6$}
\psfrag{-5}{\scriptsize\hspace{-0.15cm}$-5$}
\psfrag{-10}{\scriptsize\hspace{-0.15cm}$-10$}
\psfrag{-15}{\scriptsize\hspace{-0.15cm}$-15$}
\psfrag{-20}{\scriptsize\hspace{-0.15cm}$-20$}
\includegraphics[width=8cm]{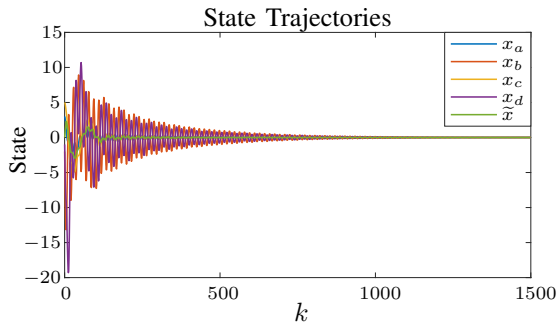}
\caption{State trajectories of the plant (\ref{eq:example plant}) and the HIGS (\ref{eq:DT-HIGS}), which are interconnected in positive feedback. Starting from nonzero initial conditions, all the state variables converge to zero. The closed-loop system is asymptotically stable, which is consistent with our expectation according to Theorem \ref{theorem:stability}.}
\label{fig:simulation}
\end{figure}

\section{CONCLUSION AND FUTURE WORK}\label{sec:conclusion}
We proposed a control framework for the digital control of linear NI systems using HIGS controllers. Discrete-time HIGS are shown to be SANI systems. For any linear discrete-time NI systems obtained via ZOH sampling, there exists a HIGS controller such that their closed-loop interconnection is asymptotically stable. An example is provided, where a discretized mass-spring system, which is NI, is stabilized using a HIGS controller.

The results presented in this paper can be generalized to multi-input multi-output (MIMO) systems by introducing a discrete-time multi-HIGS. The SANI property of a discrete-time multi-HIGS can also be investigated. Also, the stability of the closed-loop interconnection of a MIMO discrete-time NI system and a multi-HIGS can be investigated.


\bibliographystyle{IEEEtran}

\end{document}